\documentclass[prl,nobalancelastpage,twocolumn,superscriptaddress,nolongbibliography]{revtex4-2}

\usepackage{color,amsthm,amsmath,amsxtra,amsfonts,dsfont,graphicx,bm,amssymb}
\usepackage[colorlinks=true,linkcolor=blue, citecolor=blue, urlcolor=blue, bookmarks]{hyperref}
\usepackage{centernot}
\usepackage[dvipsnames]{xcolor}
\usepackage{graphicx}
\usepackage{tikz}
\usepackage{braket}
\usepackage{multirow}
\usepackage{makecell}

\def\Id{{\openone}}
\newcommand{\be}{\begin{equation}}
	\newcommand{\ee}{\end{equation}}
\newcommand{\bea}{\begin{eqnarray}}
	\newcommand{\eea}{\end{eqnarray}}
\newcommand{\bse}{\begin{subequations}}
	\newcommand{\ese}{\end{subequations}}

\theoremstyle{plain}

\newtheorem{thm}{Theorem}
\newtheorem{lem}{Lemma}
\newtheorem{defn}{Definition}

\newcommand{\tr}{\mathrm{Tr}}
\newcommand{\mc}[1]{\mathcal{#1}}

\newcommand{\prlsection}[1]{{\em {#1}.---~}}

\usepackage{pifont}
\newcommand{\xmark}{\ding{55}}%

\usepackage[export]{adjustbox}

\def\HH{\mathcal{H}}
\def\CC{\mathbb{C}}

\usepackage{tikz}
\usetikzlibrary{decorations.pathreplacing,calligraphy,decorations.markings}

\definecolor{tensor}{rgb}{1,1,1}
\definecolor{isometry}{rgb}{0.8,0.8,1}
\definecolor{unitary}{rgb}{0.8,0.5,.5}
\definecolor{gate}{rgb}{1.0,1.0,1.0}

\newcommand{\ATensor}[2]{
	\begin{scope}[shift={(#1)}]
		\draw (-1,0) -- (1,0);
		\draw (0,1) -- (0,0);
        \filldraw[fill=tensor] (0,0) circle [radius=0.5];

		\draw (0,0) node {\scriptsize #2};
	\end{scope}
}
\newcommand{\ADaggTensor}[2]{
	\begin{scope}[shift={(#1)}]
		\draw (-1,0) -- (1,0);
		\draw (0,-1) -- (0,0);
                \filldraw[fill=tensor] (0,0) circle [radius=0.5];
		\draw (0,0) node {\scriptsize #2};
	\end{scope}
}

\newcommand{\BTensor}[2]{
	\begin{scope}[shift={(#1)}]
		\draw (-1,0) -- (1,0);
		\foreach \x in {0,1,...,3}{
			\draw[shift={(-.3+0.2*\x,0)}] (0,1) -- (0,0);
		}
                \filldraw[fill=tensor] (0,0) circle [radius=0.5];
		\draw (0,0) node {\scriptsize #2};
	\end{scope}
}

\begin{document}

\title{Long-range nonstabilizerness and phases of matter}

\begin{abstract}
Long-range nonstabilizerness can be defined as the amount of nonstabilizerness which cannot be removed by shallow local quantum circuits. In this work, we study long-range nonstabilizerness in the context of many-body quantum physics, a task with possible implications for quantum-state preparation protocols and implementation of quantum-error correcting codes. After presenting a simple argument showing that long-range nonstabilizerness is a generic property of many-body states, we restrict to the class of ground states of gapped local Hamiltonians. We focus on one-dimensional systems and present rigorous results in the context of translation-invariant matrix product states (MPSs). By analyzing the fixed points of the MPS renormalization-group flow, we provide a sufficient condition for long-range nonstabilizerness, which depends entirely on the local MPS tensors. Physically, our condition captures the fact that the mutual information between distant regions of stabilizer fixed points is quantized, and this fact is not changed after applying shallow quantum circuits. We also discuss possible ramifications in the classification of phases of matter and quantum error correction.
\end{abstract}

\author{David Aram \surname{Korbany}}
\affiliation{Dipartimento di Fisica e Astronomia, Università di Bologna and INFN, Sezione di Bologna, via Irnerio 46, I-40126 Bologna, Italy}

\author{Michael J. \surname{Gullans}}
\affiliation{Joint Center for Quantum Information and Computer Science, University of Maryland and NIST, College Park, MD 20742}

\author{Lorenzo \surname{Piroli}}
\affiliation{Dipartimento di Fisica e Astronomia, Università di Bologna and INFN, Sezione di Bologna, via Irnerio 46, I-40126 Bologna, Italy}

\maketitle

\prlsection{Introduction} Stabilizer states and Clifford operations are fundamental tools in quantum information theory~\cite{gottesman1997stabilizer,nielsen2011quantum}. Clifford operators form a special set of quantum operations that generate the stabilizer states and can be simulated efficiently on classical computers~\cite{gottesman1998theory,gottesman1998heisenberg,aaronson2004improved}. Thus, stabilizer states and Clifford operations provide very useful toy models for quantum computation. In addition, due to their special properties, stabilizer states provide ideal building blocks for the construction of quantum error-correcting codes~\cite{kitaev2003fault,eastin2009restriction}.

A closely related notion is that of nonstabilizerness~\cite{bravyi2005universal,veitch2014resource} (also known as magic), which, roughly speaking, quantifies the degree to which a certain state deviates from a stabilizer state. Nonstabilizerness has received increasing attention in the past few years, also due to the recent progress in fault-tolerant quantum computing~\cite{bluvstein2024logical,acharya2024quantum}. Indeed, fault-tolerant implementation of Clifford operations is efficient and often considerably less demanding than other operations~\cite{eastin2009restriction, preskill1998fault,shor1996fault, litinski2019magic,orts2023efficient, orts2023efficient}. Quantifying nonstabilizerness is thus important from the point of view of quantum simulation: wavefunctions characterized by large nonstabilizerness will generally require more experimental resources to be prepared on a fault-tolerant quantum computer~\cite{howard2017application}. 

\begin{table}[t]
    \renewcommand{\arraystretch}{2}
	\begin{ruledtabular}
			\begin{tabular}{c |c  c } 
				& SRE   &LRE    \\
        \hline
        SRN & \hspace{0.3cm} $\ket{0}^{\otimes N}$ \hspace{0.3cm} & \hspace{0.3cm} $(\ket{0}^{\otimes N}+\ket{1}^{\otimes N})/\sqrt{2}$ \hspace{0.3cm} \\
        \hline
        LRN & \xmark & Theorem~\ref{th:main_result}\\
			\end{tabular}
		    \caption{The table shows the relation between short- and long-range entanglement (SRE and LRE, respectively) and short- and long-range nonstabilizerness (SRN and LRN, respectively) in MPSs. While shot-range nonstabilizerness is compatible with both short- and long-range entanglement, long-range nonstabilizerness cannot be realized in short-range entangled states. Our main result, Theorem~\ref{th:main_result}, provides a sufficient condition for long-range nonstabilizerness.}
    \label{tab:LRN_vs_LRE}
		\end{ruledtabular}
	\end{table}

Recently, nonstabilizerness has also been studied from a theoretical perspective in many-body physics~\cite{white2021conformal, ellison2021symmetry,sarkar2020characterization,sewell2022mana}, and much progress has been made in the problem of finding corresponding computable measures~\cite{oliviero2022magic,liu2022many,haug2022quantifying, tarabunga2023many,tarabunga2024critical,falcao2024nonstabilizerness,tarabunga2024nonstabilizerness2}. An emerging theme~\cite{white2021conformal} is that nonstabilizerness could potentially be a useful tool to characterize many-body states and phases of matter~\cite{white2021conformal,ellison2021symmetry,sarkar2020characterization,liu2022many,haug2022quantifying,tarabunga2023many,tarabunga2024critical,falcao2024nonstabilizerness}, as well as quantum dynamics~\cite{leone2022stabilizer,leone2021quantum,haferkamp2022random,haug2024probing,lopez2024exact,turkeshi2024magic,dowling2024magic}. Several works have also investigated the interplay between nonstabilizerness and other physical properties, such as entanglement~\cite{tirrito2024quantifying,gu2024magic,fux2023entanglement,frau2024nonstabilizerness,bejan2024dynamical,tarabunga2024magictransition,iannotti2025entanglement} and quantum chaos~\cite{lami2024quantum,turkeshi2023measuring,leone2021quantum,leone2023nonstabilizerness,leone2023phase,garcia2023resource,turkeshi2025pauli,bera2025non}.

An interesting development has been the introduction of long-range (LR) nonstabilizerness~\cite{white2021conformal,ellison2021symmetry}. Originally, it was studied in the context of conformal field theories~\cite{white2021conformal} where it was defined in terms of certain long-range correlations, see also Ref.~\cite{sarkar2020characterization}. It was shown in particular that such long-range correlations are non-vanishing in the critical ground state of a $q=3$ Potts model, thus displaying LR nonstabilizerness. Similar non-local correlation functions, written in terms of the so-called stabilizer R\'enyi entropy~\cite{leone2022stabilizer}, were later studied in different critical spin chains~\cite{tarabunga2023many,frau2024stabilizer} and quantum dynamics~\cite{lopez2024exact,tarabunga2024magic}.

We can give a more transparent definition of LR nonstabilizerness, in analogy with the notion of long-range entanglement~\cite{white2021conformal,ellison2021symmetry,haug2022quantifying}. Specifically, LR nonstabilizerness can be defined as the amount of nonstabilizerness which cannot be removed by shallow local quantum circuits. This definition is useful from the point of view of quantum simulation. Indeed, in many-body physics, one is often interested in long-range correlations, which are not changed by local unitary transformations such as shallow quantum circuits. Therefore, when a target state does not have LR nonstabilizerness, it is possible to find another state with the same long-range correlations but which is a stabilizer state and thus easier to implement fault tolerantly. Additionally, one can define symmetry-protected LR nonstabilizerness~\cite{ellison2021symmetry} by imposing a symmetry on the circuits and the states. As an interesting result, it was proven that certain symmetry-protected topological (SPT) phases necessarily display symmetry-protected LR nonstabilizerness~\cite{ellison2021symmetry}.

Despite this recent work, the notion of LR nonstabilizerness remains largely unexplored. For instance, in the context of ground-state physics, an important open question pertains to the possibility of finding signatures of LR nonstabilizerness in a state wavefunction. This question is naturally inspired by the established characterization of certain long-range entangled states, which are known to display a non-zero topological entanglement entropy~\cite{hamma2005ground,levin2006detecting, kitaev2006topological,zeng2015quantum,zeng2016topological,fromholz2020entanglement}. 

In this work, we formalize and tackle this problem. After presenting a simple argument showing that LR nonstabilizerness is a generic property of many-body states, we provide rigorous results in the simplest case of one-dimensional $(1D)$ systems described by translation-invariant (TI) matrix product states (MPSs)~\cite{perez2007matrix,cirac2017matrix,cirac2020matrix}. By analyzing the fixed points of the MPS renormalization-group
flow, we provide a sufficient condition for LR nonstabilizerness that depends entirely on the local MPS tensors. In addition to shedding further light into the structure of nonstabilizerness in many-body wavefunctions, our work may have ramifications in the context of quantum error correction and in the preparation of encoded states of stabilizer codes.

\prlsection{Long-range nonstabilizerness} We begin by introducing Clifford operations and stabilizer states~\cite{nielsen2011quantum}. We consider a system of $N$ qubits associated with the Hilbert space $\mathcal{H}=\otimes_{j=1}^N\mathcal{H}_j$, where $\mathcal{H}_j\simeq \mathbb{C}^2$. We denote by $\sigma_j^\alpha$ the Pauli matrices acting on site $j$, with $\alpha=0,1,2,3$ ($\sigma^0=\openone$ being the identity) and by $\{|0\rangle, |1\rangle\}$ the local computational basis. We also denote by $\mathcal{P}_N$ the set of all $N$-qubit Pauli strings including a global phase $\phi\in [\pm 1,\pm i]$. The Clifford group is the set of
unitaries $U$ such that $UPU^\dagger \in \mathcal{P}_N$ for all $P \in \mathcal{P}_N$. The pure stabilizer states are the states
generated by applying elements of the Clifford group to the reference state $\ket{0}^{\otimes N}$.

Since we will be interested in the thermodynamic limit, we will consider sequences of states $\{\ket{\Psi_N}\in \mathcal{H}_N\}_{N \in \mathbb{N}}$ defined on systems of increasing size $N$. Assuming that the qubits are arranged over a regular lattice of dimension $D$ (we will mostly focus on $D=1$), we can define the family of local quantum circuits (QCs) as the unitaries $ Q_\ell =  V_{\ell} \ldots V_2 V_1$, where each ``layer'' $V_{n}$ contains quantum gates acting on disjoint sets of two neighboring qubits (\emph{i.e.} we require local gates). The integer $\ell$ is the depth of the circuit. We will not restrict the gates to a given gate set. We are now in a position to define LR nonstabilizerness.

\begin{defn}\label{def:sm}
A family of states $\{\ket{\psi_N}\}_{N \in \mathbb{N}}$, has short-range (SR) nonstabilizerness if for all $\varepsilon_0>0$ and $\alpha>0$, there exists a local QC $Q_{D_N}$ of depth $D_N =O( {\rm polylog}(N))$ and a stabilizer state $\ket{S_N}$ such that for large enough $N$
\be \label{eq:sm}
\Delta(Q_{D_N} \ket{\psi_N},\ket{S_N})\leq  \frac{\varepsilon_0}{ N^\alpha} = \varepsilon_N,
\ee
where $\Delta(\ket{\psi},\ket{\phi})=\sqrt{1-|\braket{\psi|\phi}|^2}$ is the trace distance. If the sequence $\{\ket{\psi_N}\}_{N \in \mathbb{N}}$ does not have SR nonstabilizerness, we say that it has LR nonstabilizerness.
\end{defn}
Compared to previous definitions~\cite{ellison2021symmetry}, Eq.~\eqref{eq:sm} allows for some error $\varepsilon_N$ that vanishes (in a mild way) in the thermodynamic limit. Allowing for a small error is needed in order to capture states with correlations decaying exponentially with the distance (shallow QCs cannot introduce exponential tails). This definition is analogous to that of LR entanglement, with the difference that a long-range entangled state cannot be mapped to a product state by a shallow circuit $Q_N$. Since any product state is locally equivalent to the stabilizer state $\ket{0}^{\otimes N}$, LR nonstabilizerness is stronger than LR entanglement. In fact, LR nonstabilizerness is strictly stronger, in the sense that there are states with LR entanglement but not LR nonstabilizerness, cf. Table~\ref{tab:LRN_vs_LRE}. We also note that Def.~\ref{def:sm} differs from the one given in Refs.~\cite{cao2024gravitational,cao2024non,qian2025quantum}. For a bipartition of the system into regions $A$ and $B$, these references considered unitary operations of the form $U_A \otimes U_B$ and not QCs (which alter the entanglement between $A$ and $B$), leading to a notion of non-local nonstabilizerness which is of interest in different contexts. It is also worth repeating that Def.~\ref{def:sm} is a statement about the thermodynamic limit.

\prlsection{Typicality of LR nonstabilizerness} The first natural question is whether LR nonstabilizerness exists at all. We give an heuristic counting argument showing that, in fact, LR nonstabilizerness is a \emph{typical} property of many-body quantum states. To this end, we count the number of quantum states in a Hilbert space of dimension $\mathcal{D}=2^N$ that can be distinguished up to an error $\varepsilon_N$. This is equivalent to the number of $\varepsilon_N$-balls in the Hilbert space, scaling as~\cite{mele2024introduction} $n_B\sim e^{(1-\varepsilon^2_N) 2^{N-1}}$. Next, we count the number of stabilizer states to which we apply a circuit of depth $D_N$. First, we recall that the number of stabilizer states $n_S$ scales as~\cite{singal2023counting} $n_S\sim 2^{N^2/2}$. Second, we count the number of circuits of depth $D_N$. To this end, we recall that a quantum circuit of $m$ two-qubit gates can be approximated to an error $\varepsilon_N$ (in operator norm) by a quantum circuit of $O(m \operatorname{polylog}(m/\varepsilon_N))$ gates chosen from a universal gate set~\cite{nielsen2011quantum} of a finite number of elements, which we denote by $n_g$. Since the total number of space-time positions for a gate is $\sim ND_N \operatorname{polylog}(N D_N/\varepsilon_N)$, and in each position we can choose $n_g$ gates, the total number of such local quantum circuits is $n_C\sim e^{N D_N\operatorname{polylog}(N D_N/\varepsilon_N)}$. In conclusion, for $D_N\sim {\rm polylog}(N)$, $\varepsilon_N=\varepsilon_0/N^\alpha$, we have
\begin{equation}
\lim_{N\to \infty}\frac{ n_C n_S}{n_B}\to 0\,,
\end{equation}
meaning that LR nonstabilizerness is typical in many-body states. This is the first main result of our work.

This result is analogous to the statement that LR entanglement is a typical property of many-body states. In that case, typicality of LR entanglement is due to the well-known fact that most states in the many-body Hilbert space display an extensive bipartite entanglement entropy~\cite{page1993average}, which cannot be removed, not even approximately, by shallow QCs~\cite{piroli2020quantum}.

\prlsection{LR nonstabilizerness and phases of matter} After establishing the typicality of LR nonstabilizerness, we turn to the more interesting problem of characterizing LR nonstabilizerness in the very important class of ground states of local gapped Hamiltonians. Once again, a strong motivation to tackle this problem comes from the study of entanglement. Indeed, the characterization of LR entanglement in ground states of local gapped Hamiltonians has played a pivotal role in our understanding of topological phases of matter~\cite{chen2010local,hastings2013classifying,zeng2015quantum,zeng2015gapped,chiu2016classification}, revealing important hidden structure of the corresponding wavefunctions.

As anticipated, we will consider the simplest case of $1D$ systems and provide rigorous results in the context of MPSs. This is not a restriction, as MPSs have been shown to approximate arbitrarily well ground states of $1D$ local Hamiltonians~\cite{verstraete2006matrix,hastings2007area,arad2013area,huang2014area}, providing at the same time useful toy models for analytic inspection~\cite{cirac2017matrix,cirac2020matrix}. 

As we are interested in describing quantum phases in the thermodynamic limit, we will focus on TI MPSs. For a lattice of local Hilbert-space dimension $d$, MPSs are defined by a single tensor $A^i_{\alpha,\beta}$, where $i=1,\dots , d$, $\alpha,\beta = 1,\dots ,\chi$, their wavefunction taking the form
\be
\ket{v^{(N)}(A)}=\sum_{i_1, \ldots, i_N=1}^d \operatorname{tr}\left(A^{i_1} \ldots A^{i_N}\right)\left|i_1 \dots i_N\right\rangle\,.
\ee
Here we view $A^i$ as a $\chi$-by-$\chi$ matrix, where $\chi$ is called the bond dimension.

Before proceeding, it is necessary to recall some facts about MPSs. We first introduce the transfer matrix
\begin{equation}
	\mathbb{E}_A = \sum_{i=1}^d A^{i} \otimes A^{\ast i}\,.
\label{eq:transfer_matrix}
\end{equation}
Next, we recall two key concepts: the blocking procedure and the MPS canonical form (CF)~\cite{cirac2017matrix}. Given an MPS with tensor $A^i$, one can construct a new tensor by grouping (blocking) $q$ sites, such that the corresponding matrices are obtained by the product $B^i=A^{i_1}\cdots A^{i_q}$, where $i$ denotes all possible sets of indices $\{i_1,\ldots i_q\}$. The new tensor $B$ has physical dimension $d^q$, the same bond dimension $\chi$ and transfer matrix $\mathbb{E}_B=\mathbb{E}_A^q$.

In order to introduce the CF, we first define normal tensors. A tensor $A$ is called normal, if $(i)$ it is irreducible, \emph{i.e.}  the $A^i$ have
no common nontrivial invariant subspace, and $(ii)$ $\mathbb{E}_A $
has a unique largest eigenvalue $\lambda_1=1$ and no other eigenvalue with the same magnitude. Physically, MPSs defined by a normal tensor display short-range correlation functions (and thus SR entanglement), where the correlation length $\xi$ is determined by the first subleading eigenvalue  $\lambda_2(\mathbb{E}_A)$ via $\xi=-1/\ln |\lambda_2|$.

A fundamental result is that, after blocking and possibly up to a gauge transformation $A^{i}\to W A^{i} W^{-1}$~\cite{cirac2017matrix}, every tensor $A^{i}$ can be expressed in terms of normal tensors in the CF $A^i=\bigoplus_{j=1}^b \operatorname{diag}\left(\mu_{j, 1}, \ldots, \mu_{j, m_j}\right) \otimes A_j^i$, where $\mu_{j,k}$ are complex numbers, with $|\mu_{j,k}|\leq 1$ and at least
one of them having magnitude exactly one. Here, $A^{i}_j$ are normal tensors generating MPSs that are mutually orthogonal in the thermodynamic limit. It follows that any normalized TI MPS can be written as
\begin{equation}\label{eq:CF}
    \left|\phi_N\right\rangle=\frac{1}{c_N} \sum_{j=1}^b \beta^{(N)}_j\ket{v^{(N)}(A_j)}\,,
\end{equation}
where $\beta^{(N)}_j=\sum_{k=1}^{m_j} \mu_{j, k}^N$, while $c_N$ is the normalization. States of the form~\eqref{eq:CF} with $b>1$ may represent long-range entangled states, a typical example being the superposition of symmetry-broken ground states in a degenerate ground space of a symmetric local Hamiltonian.

\prlsection{RG flow and LR nonstabilizerness} A powerful tool in the MPS framework is the possibility to define a real-space renormalization-group (RG) flow~\cite{cirac2017matrix}. The RG iteration consists in blocking $2$ sites, $A^i\mapsto B^i=A^{i_1} A^{i_2}$, followed by a polar decomposition $B=VC$, where V is an isometry $V^\dagger V=\openone_{\chi^2}$. The RG flow maps the tensor $A$ into $C$, implementing a coarse-graining procedure. 

The form of RG fixed points is known~\cite{cirac2017matrix}. In order to define them, consider a system of $N$ sites, each hosting a left and right qudit $L_n$, $R_n$. Then, up to local isometries acting on adjacent sites $L_n,R_n$, they read~\cite{cirac2017matrix}
\begin{equation}\label{eq:fixed_points}
    |\tilde{\phi}_N\rangle=\sum_{j=1}^b \alpha_j^{(N)}\ket{\Omega_j} \,,
\end{equation}
where $\ket{\Omega_j} =\bigotimes_{i=1}^{N} \ket{\omega_j}_{R_i ,L_{i+1}}$ ($\alpha^{(N)}_j$ may depend on $N$). Here $\ket{\omega_j}_{R_i ,L_{i+1}}=\sum_{m_j=1}^{d_j} \sqrt{\lambda_{m_j}}\ket{m_j}_{R_{i}}\otimes \ket{m_{j}}_{L_{i+1}}$, with $\braket{m_j|m_{j^\prime}}_{R_i ,L_{i+1}}=\delta_{j,j'}$
for all $i$. This condition states the local orthogonality of the RG fixed points.

RG fixed points have a simple form, allowing us to make progress on the characterization of LR nonstabilizerness. To this end, we make use of recent results constructing efficient preparation-protocols for MPSs~\cite{piroli2021quantum,malz2024preparation}. Given an MPS $\ket{\phi_N}$, which we can take to be in CF as in ~\eqref{eq:CF}, and any $\varepsilon_N>0$ (possibly depending on $N$), it was shown in Ref.~\cite{malz2024preparation} that there exists an RG fixed-point MPS $\ket{\tilde{\phi}_N}$ and a QC of depth $D_N=O(\log(N/\varepsilon_N))$ such that
$\Delta(\ket{\phi_N},Q_{D_N}\ket{\tilde{\phi}_N})\leq \varepsilon_N$. 

Note that $\ket{\tilde{\phi}_N}$ can be entirely determined based on the knowledge of the local tensors of $\ket{\phi_N}$. Hence, since any MPS is related to an MPS in the form~\eqref{eq:fixed_points} by a shallow QC, Def.~\ref{def:sm} implies that an MPS $\ket{\phi_N}$ has LR nonstabilizerness if and only the corresponding state $\ket{\tilde{\phi}_N}$ does~\footnote{The RG fixed points are typically defined on chains of qudits of dimension $d$. However, we can always view a local $d$-dimensional space as that of $\lceil \log_2(d)\rceil$ qubits. Note that any shallow circuit made of $2$-qudit gates on the qudit chain can always be expressed as a shallow local circuit made of $2$-qubit gates in the associated qubit chain, possibly up to a polylogarithmic overhead.}. We can therefore restrict to studying LR nonstabilizerness in the class of RG fixed points.

\prlsection{A sufficient condition for LR nonstabilizerness} We are left with the task of characterizing LR nonstabilizerness of RG fixed points. Clearly, LR nonstabilizerness requires $b>1$ in Eq.~\eqref{eq:fixed_points}, otherwise the state has SR entanglement and thus SR nonstabilizerness, cf. Table~\ref{tab:LRN_vs_LRE}. When $b>1$, the problem of characterizing LR nonstabilizerness becomes non-trivial. To see this, consider RG fixed points of the form 
\begin{equation}
\label{eq:ghz}
|\tilde{\phi}_N[\alpha,\beta]\rangle=\alpha|0\rangle^{\otimes N}+\beta|1\rangle^{\otimes N}\,,
\end{equation}
where we take $\alpha$ and $\beta$ to be $N$-independent and with $ |\alpha|^2 + |\beta|^2 = 1$~\footnote{One could ask if $|\tilde{\phi}[\alpha,\beta]\rangle$ can be represented as a TI MPS (with $N$-independent tensors and without boundary tensors), for any choice of $\alpha,\beta\in \mathbb{C}$, or at least for a dense set of them. This is indeed the case, possibly restricting to some values of $N$. For example, by summing $n$ product states $\ket{0}^{\otimes N}$ and $m$ product states $\ket{1}^{\otimes N}$, we obtain a TI MPS of the form~\eqref{eq:ghz} with $\alpha=n/\sqrt{n^2+m^2}$ and $\beta=m/\sqrt{n^2+m^2}$, which can approximate any pair of positive real numbers. When $\alpha$ is complex (we can always assume that $\beta\in \mathbb{R}^+$ up to a global phase), we can sum product states defined by local tensors with an additional phase $e^{i\varphi}$, yielding the state $e^{i N\varphi}\ket{0}^{\otimes N}$. Choosing $\varphi=2\pi r/s$  with integers $r, s$ and restricting to system sizes $N\equiv 1$ (mod $s$), we obtain states of the form~\eqref{eq:ghz} where $\alpha$ has a phase angle $\varphi=2\pi r/s$\label{footnote}}. When $\alpha=0, 1$, $\ket{\tilde{\phi}[\alpha,\beta]}$ is a stabilizer state, while for $\alpha= 2^{-1/2}$ it is easy to show that $\ket{\tilde{\phi}[\alpha,\beta]}$ has SR nonstabilizerness. Indeed, in this case $\beta=e^{i\varphi}/\sqrt{2}$ for some $\varphi$. By applying the single-qubit gate $U_1=e^{i\varphi (\sigma^z_1-\openone_1)/2}$ on the first site, we obtain the stabilizer state $(\ket{0}^{\otimes N}+\ket{1}^{\otimes N})/\sqrt{2}$. Therefore, $\ket{\tilde{\phi}[1/\sqrt{2},\beta]}$ has SR nonstabilizerness. However, it is not clear if $\ket{\tilde{\phi}[\alpha,\beta]}$ has LR nonstabilizerness for $\alpha\neq 0, 1, 2^{-1/2}$. This is indeed the case, as follows from our second main result.
\begin{thm}\label{th:main_result}
A sufficient condition for an MPS to have LR nonstabilizerness, according to Def.~\ref{def:sm}, is that its RG fixed point~\eqref{eq:fixed_points} satisfies
\begin{equation}\label{eq:main_eq}
    \lim_{N\to \infty} H(\{|\alpha^{(N)}_j|^2\})\notin \mathbb{N}
\end{equation}
where we introduced the Shannon entropy $H(\{p_j\}) = -\sum_j p_j \log_2(p_j)$ and assumed $\sum_j|\alpha^{(N)}_j|^2=1$.
\end{thm}

\begin{figure}
    \centering
    \includegraphics[scale=0.35]{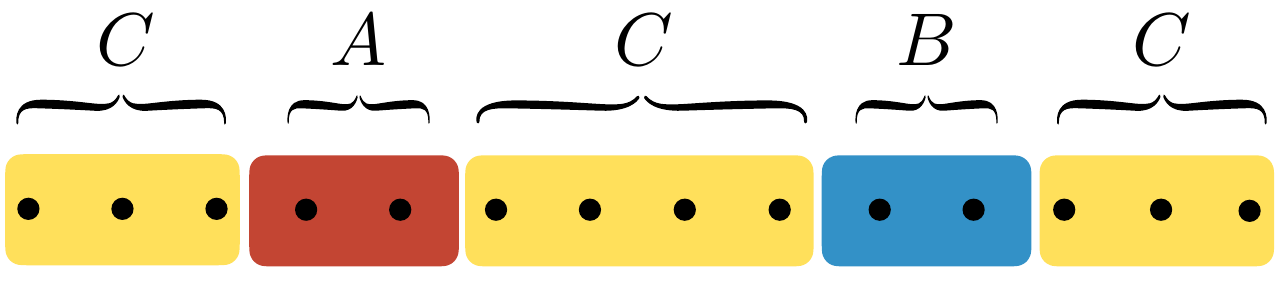}
    \caption{Partition considered in the proof of Theorem~\ref{th:main_result}. The $1D$ periodic chain is
divided into three disjoint regions $A$, $B$, and $C$, where $A$ and $B$ are sufficiently separated intervals, while $C$ is the complement of $A\cup B$.}
    \label{fig:partition}
\end{figure}

Theorem~\ref{th:main_result} relies on the special entanglement features of stabilizer states~\cite{hamma2005bipartite,hamma2005ground}. Let $\ket{\tilde{\phi}_N}$ be a RG fixed point of the form~\eqref{eq:fixed_points} and assume that there exists $Q_{D_N}$ with $D_N=O(\operatorname{polylog}(N))$ such that
\be\label{eq:assumption_contraditction}
\Delta(Q_{D_N}\ket{\tilde{\phi}_N},\ket{S_N}) \leq  \varepsilon_N\,,
\ee
where $\ket{S_N}$ are stabilizer states. Now, consider a partition of the system into three regions, $A$, $B$ $C$, where $A$ and $B$ are disconnected intervals, while $C$ is the complement of $A\cup B$, cf. Fig.~\ref{fig:partition}. We take the size of each connected component of $C$  to be $O(N)$, and $|A|=|B|$ with 
\begin{equation}\label{eq:size_AB}
     4D_N+4\leq   |A\cup B| \leq 8 D_N\,.
\end{equation}
For any state $\ket{\psi}$, we define the mutual information~\cite{nielsen2011quantum} $I_{A,B}[\psi]=S(\rho_A)+S(\rho_B)-S(\rho_{AB})$ where $\rho_{R}={\rm Tr}_{\bar{R}}(\ket{\psi}\bra{\psi})$ ($\bar{R}$ is the complement of the region $R$), while $S(\rho)=-{\rm Tr}[\rho\log_2 \rho]$ is the von Neumann entropy. The proof of Theorem~\ref{th:main_result} requires the following facts, which are rigorously stated and proved in Ref.~\cite{SM}:
\begin{enumerate}
\item \label{point_1}The mutual information of the state $\ket{\psi_N}=Q_{D_N}\ket{\tilde{\phi}_N}$ is $I_{A,B}[\psi_N]=H[\{|\alpha_j^{(N)}|^2\}]$ where $\alpha_j^{(N)}$ are the coefficients appearing in Eq.~\eqref{eq:fixed_points}. This is a consequence of the fact that, for suitably chosen partitions $A\cup C \cup B$, the mutual information  of RG fixed points cannot be changed by shallow circuits;
\item \label{point_2} The mutual information of any stabilizer state $\ket{S_N}$ is quantized, $I_{A,B}[S_N]=r\,, r\in \mathbb{N}$;
\item \label{point_3} $I_{A,B}[\nu]$ is a continuous function of the state $\ket{\nu}$, which follows from the Fannes inequality~\cite{audenaert2007sharp}: given two states $\rho_R$ and $\sigma_R$ in a region $R$, with trace distance $\delta = ||\rho-\sigma ||_1$/2, we have 
\begin{equation}\label{eq:fannes}
    |S(\rho_R) - S(\sigma_R)|\leq \delta |R| + H_{\rm bin}(\delta)\,,
\end{equation}
where $H_{\rm bin}(\varepsilon)=-\varepsilon\log_2 \varepsilon -(1-\varepsilon)\log_2 (1-\varepsilon)$ is the binary entropy.
\end{enumerate}
Now, Eq.~\eqref{eq:assumption_contraditction} implies that the trace distance between the reduced density matrices of $\ket{\psi_N}$ and of $\ket{S_N}$ on any region $R$ vanishes in the thermodynamic limit, due to the contractivity of the trace distance~\cite{nielsen2011quantum}. Then, exploiting the bound~\eqref{eq:fannes} and the fact that $\varepsilon_N|A\cup B|$ also vanishes for $N\to\infty$ [thanks to Eq.~\eqref{eq:size_AB}], it can be shown that the mutual information of $\ket{\psi_N}$ converges to that of $\ket{S_N}$. Using now the facts~\ref{point_1} and \ref{point_2} above, we finally establish Eq.~\eqref{eq:main_eq}, completing the proof~\cite{SM}. 

Restricting to Eq.~\eqref{eq:ghz}, Theorem~\ref{th:main_result} allows us to completely classify the values of $\alpha,\beta$ for which the state has LR nonstabilizerness. Indeed, $H_{\rm bin}[|\alpha|^2] \in \mathbb{N}$ if and only if $|\alpha|^2=0,1,1/2$. Therefore, when restricting to MPSs flowing to the RG fixed points~\eqref{eq:ghz}, Eq.~\eqref{eq:main_eq} is a necessary and sufficient condition for LR nonstabilizerness. However, we expect that this is not always the case. As an example, consider the state 
\begin{align}
\label{eq:counter_example}
   \! \ket{\tilde{\phi}_N(t)}=&\alpha_1(t)\ket{00}^{\otimes N/2}+\alpha_2(t)\ket{01}^{\otimes N/2}\nonumber\\
   +&\alpha_3(t)\ket{10}^{\otimes N/2}+\alpha_4(t)\ket{11}^{\otimes N/2}\,,
\end{align}
where $\alpha^2_1(t)=10^{-1}$, $\alpha^2_2(t)=3^{-1/4}10^{-1}$, $\alpha^2_3(t)=t$ and $\alpha_4(t)^2=1-t-10^{-1}-3^{-1/4}10^{-1}$. By inspection, we see that there is a value $t^\ast\simeq 0.023$ for which $H[\{|\alpha_j(t^\ast)|^2\}]=1$. In this case, despite the mutual information being an integer, we conjecture that $\ket{\tilde{\phi}_N(t^\ast)}$ has LR nonstabilizerness. Although we cannot prove it, we can show a weaker statement, namely that there is no QC $Q_{D_N}$ with $D_N=O(\operatorname{polylog}(N))$ such that $Q_{D_{N}}\ket{\tilde{\phi}_N(t^\ast)}=\ket{S_N}$, where $\ket{S_N}$ is a stabilizer state. This latter condition corresponds to an exact notion of SR nonstabilizerness, where we set $\varepsilon_0=0$ in Eq.~\eqref{eq:sm}. This statement is a consequence of the following result.
\begin{thm}\label{th:main_result_2}
A necessary condition for the RG fixed point~\eqref{eq:fixed_points} to have exact SR nonstabilizerness [i.e. to satisfy Eq.~\eqref{eq:sm} with $\varepsilon_0=0$] is that, for all $i\neq j$,
\begin{equation}\label{eq:main_result_2}
    |\alpha_i|^4/|\alpha_j|^4 \in \mathbb{Q}\,.
\end{equation}
\end{thm}
\noindent The proof relies on the fact that, for any stabilizer state $\ket{S}$ with reduced density matrix $\rho_{A,B}=\operatorname{Tr}_C[\ket{S}\bra{S}]$, it must be $(\rho_{A,B}^{T_A})^2\propto (\rho_{A,B}^{T_A})^4$~\cite{sang2021entanglement}, where $(\cdot)^{T_A}$ denotes partial transpose over the region $A$. Using this property, we show in Ref.~\cite{SM} that $|\alpha_i(t^\ast)|^4/|\alpha_j(t^\ast)|^4 \in \mathbb{Q}$ for $i\neq j$ is a necessary condition for $Q_{D_N} \ket{\tilde{\phi}_N(t^\ast)} = \ket{S_N}$. 

Since Eq.~\eqref{eq:main_result_2} is not satisfied by the state~\eqref{eq:counter_example} for the coefficients we have chosen, we conclude that there is no shallow QC mapping $\ket{\tilde{\phi}_N(t^\ast)}$ exactly into a stabilizer state. Going beyond, we should prove that there is no QC for which the approximate equality Eq.~\eqref{eq:sm} holds. Technically, this appears to be challenging, as the trace distance generally increases exponentially in system size under partial transpose~\cite{lu2020entanglement}, making it difficult to extend the proof that we used for mutual information~\cite{SM}. 

\prlsection{Outlook} Our work raises several questions. First, while our sufficient condition for LR nonstabilizerness in MPSs appears to be quite powerful, it would be important to find a condition that is also necessary, which would allow us to resolve, for instance, our conjecture about the LR nonstabilizerness of states such as~\eqref{eq:counter_example}. Second, it is natural to ask whether one could rigorously establish a connection between LR nonstabilizerness, as given in Def.~\ref{def:sm}, and the behavior of certain non-local correlation functions, \emph{e.g.} as encoded in the mutual Stabilizer Rényi entropies~\cite{tarabunga2023many}. Here, a non-trivial challenge is the derivation of suitable continuity bounds, which are needed to account for the error $\varepsilon_N$ in Def.~\ref{def:sm}. Finally, it is interesting to ask whether states with SR nonstabilizerness can be viewed as the trivial phase with respect to a suitable equivalence relation, and study the corresponding equivalence classes.

Let us also discuss some possible future directions. An obvious extension of our results would be the characterization of long-range SPT nonstabilizerness~\cite{ellison2021symmetry} in the context of MPSs. Here, we expect that progress could be made by combining some of the ideas introduced in our work with classical results on long-range SPT entanglement of MPSs~\cite{pollmann2010entanglement,chen2011classification,schuch2011classifying}. It is also very natural to ask whether one can define and study LR nonstabilizerness of unitary operators, paralleling the topological classification of quantum cellular automata~\cite{gross2012index,cirac2017MPUs,sahinoglu2018matrix}, or extend the notion of LR nonstabilizerness by allowing for additional operations, such as local measurements and non-local classical communication~\cite{piroli2021quantum,tantivasadakarn2023hierarchy}.

Finally, while we only considered $1D$ systems, we expect a much richer picture in higher dimensions. For instance, since stabilizer Hamiltonians cannot realize non-abelian order~\cite{potter2016symmetry}, it would be interesting to study signatures of LR nonstabilizerness in the ground-state wavefunctions of non-abelian topologically-ordered models. We hope that our work will motivate further research in this direction. 

\prlsection{Acknowledgments}
We thank Marcello Dalmonte, Tyler D. Ellison, Domink Hangleiter, Yaodong Li, David T. Stephen, Ruben Verresen, and  Dominic J. Williamson for inspiring discussions, and Marcello Dalmonte, Tobias Haug, David T. Stephen, and Gerorgios Styliaris for valuable comments on the manuscript. LP and MJG acknowledge hospitality from the Simons Center for Geometry and Physics, Stony Brook University, during the “Workshop on Quantum information dynamics and non-equilibrium quantum matter'', where part of the research for this paper was performed. MJG acknowledges support from the NSF QLCI  grant OMA-2120757. This work was co-funded by the European Union (ERC, QUANTHEM, 101114881). Views and opinions expressed are however those of the author(s) only and do not necessarily reflect those of the European Union or the European Research Council Executive Agency. Neither the European Union nor the granting authority can be held responsible for them.


	\let\oldaddcontentsline\addcontentsline
	\renewcommand{\addcontentsline}[3]{}
	\bibliography{bib}
	\let\addcontentsline\oldaddcontentsline

\onecolumngrid
\newpage

\appendix
\setcounter{equation}{0}
\setcounter{figure}{0}
\renewcommand{\thetable}{S\arabic{table}}
\renewcommand{\theequation}{S\thesection.\arabic{equation}}

\setcounter{defn}{0}
\setcounter{thm}{0}
\setcounter{figure}{0}

\setcounter{secnumdepth}{2}

\begin{center}
    {\Large \bf Supplemental Material}
\end{center}

Here we provide additional details on the results stated in the main text.

\tableofcontents

\section{Stabilizer states}\label{sec:stabilizer states}

In this section, we review the main properties of stabilizer states and Clifford operators~\citep{aaronson2004improved,nielsen2011quantum}.

We consider a Hilbert space $\HH_N = \otimes_{i\in \Lambda} \HH_i $ 
on a one-dimensional lattice $\Lambda$ of length $|\Lambda| = N$ with local Hilbert spaces $\HH_i = \CC^2$, \emph{i.e.} $N$ qubits. We denote the Pauli-matrices as $\sigma^{\alpha}$, defined as 
\be
\sigma^{0} = \begin{pmatrix}
1 & 0 \\
0 & 1    
\end{pmatrix}\,,\quad
\sigma^1=\begin{pmatrix}
0 & 1 \\
1 & 0
\end{pmatrix}\,,\quad 
\sigma^2=\begin{pmatrix}
0 & -i \\
i & 0
\end{pmatrix}\,, \quad
\sigma^3=\begin{pmatrix}
1 & 0 \\
0 & -1
\end{pmatrix}.
\ee

First, we define the Pauli group $\mathcal{P}(\Lambda)$ as the finite group of the tensor products of Pauli matrices  on $\Lambda$ with a global phase $\phi \in [\pm 1, \pm i]$, i.e.\ an element  $g  \in \mathcal{P}(\Lambda)$ is 
\be
g = \phi \sigma^{\alpha_1}_1 \otimes \sigma^{\alpha_2}_2 \otimes \dots \otimes \sigma^{\alpha_N}_N,\quad \phi = \pm 1,\pm i.
\ee
The stabilizer group $G =\mathrm{Stab}(\ket{\psi})$ of a state $\ket{\psi} \in \HH_N $ is

\begin{align}
\mathrm{Stab}(\ket{\psi}) &= \{P\in \mathcal{P}(\Lambda):P\ket{\psi} = \ket{\psi}\}\,.
\end{align}
Note that elements of $G = \mathrm{Stab}(\ket{\psi})$ can only have global phases $\pm 1$. Furthermore, $G$ is an abelian subgroup of $\mc{P}(\Lambda)$.

A state $\ket{S}$ is a stabilizer state if and only if $
|\mathrm{Stab}(\ket{S})| = 2^N$, and in this case we can write~\cite{klappenecker2002beyond}
\be\label{eq:stab_state}
\ket{S} \bra{S} = \frac{1}{2^N}\sum_{g \in \mathrm{Stab}(\ket{S})} g\,,
\ee
In addition, the reduced density matrix arising from a stabilizer state $\rho_R = \tr_{\bar{R}} \ket{S} \bra{S}$ is of the form
\be\label{eq:reduced_density_matrix}
\rho_R=\frac{1}{2^{|R|} }\sum_{h \in H_R} h\,,
\ee
where $H_R$ is an abelian subgroup of $\mathcal{P}(R)$. $H_R$ contains the Paulis for which the trace in Eq.~\eqref{eq:stab_state} is nonzero, namely those of the form $g = h_R \otimes \openone_{\bar{R}}$.

Finally, we recall the following well known result~\cite{hamma2005bipartite}
\begin{lem}\label{lemma:stab_entanglement_spectrum}
The von Neumann entanglement entropy $S(\rho) = -\tr[\rho \log_2(\rho)]$ of the stabilizer state~\eqref{eq:reduced_density_matrix} reads
\begin{equation}
    S(\rho_R)=|R|-\log_{2}|H_R|\,.
\end{equation}
where we denoted by $|R|$ the number of qubits in $R$ and by $|H_R|$ the size of the subgroup $H_R$. In addition, $|H_R|=2^{\ell}$ for some integer $\ell$, so that $S(\rho_R)$ is an integer.
\end{lem}

\section{Canonical forms for MPSs}
\label{sec:canonical forms}
We briefly review the canonical forms for TI MPS and the renormalization group (RG) fixed points introduced in \cite{verstraete2005renormalization}, following the presentation in \citep{cirac2017matrix}.

A tensor $A^i_{\alpha,\beta}$, $i=1,\dots d$, $\alpha,\beta = 1,\ldots ,\chi$, generates a familiy $\mathcal{V}(A)=\{\ket{v^{(N)}(A)}\}_{N \in \mathbb{N}}$ of translation invariant MPSs
\be
\ket{v^{(N)}(A)} =\sum_{i_1, \ldots, i_N=1}^d \operatorname{tr}\left(A^{i_1} \ldots A^{i_N}\right)\left|i_1\right\rangle \otimes \ldots \otimes\left|i_N\right\rangle \in \CC_d^{\otimes N}\,.
\ee
Here we view the $A^i$ as matrices. Note that we do not consider normalized MPS. We will make use of a standard graphical notation and identify $A^{i}_{\alpha,\beta}=
\begin{array}{c}
    \begin{tikzpicture}[scale=.4, baseline={([yshift=-5.5ex]current bounding box.center)}, thick]
    	\ATensor{0,0}{$A$}
    	\draw (0,1.35) node {\scriptsize $i$};
    	\draw (1.35,0) node {\scriptsize $\beta$};
    	\draw (-1.35,0) node {\scriptsize $\alpha$};
    \end{tikzpicture}
\end{array}$, where each leg corresponds to an index. The associated completely positive map (CPM) is
\be \label{eq:cpm}
\mathcal{E}_{A}(X)=\sum_{i=1}^d A^i X A^{i \dagger} = 
    \begin{array}{c}
		\begin{tikzpicture}[scale=.5,thick,baseline={([yshift=1ex]current bounding box.center)}]
			\ATensor{0,0}{$A$}
			\ADaggTensor{0,1.7}{$A^*$}
            \draw[shift={(2, 0)}] (-1, 1.7) -- (-0.5, 1.7) -- (-0.5, 0) -- (-1, 0);
			\filldraw[color=black, fill=white, thick](1.5, 0.85) circle (0.5);
			\draw (1.5, 0.85  ) node {\scriptsize X};
		\end{tikzpicture}
	\end{array}
\ee
The transfer matrix  $\mathbb{E}_A$ is the matrix representation of the CPM Eq.\eqref{eq:cpm}:
\be
	\mathbb{E}_A = \sum_{i=1}^d A^{i} \otimes A^{\ast i}
	=
  	\begin{array}{c}
		\begin{tikzpicture}[scale=.5,thick,baseline={([yshift=1ex]current bounding box.center)}]
			\ATensor{0,0}{$A$}
			\ADaggTensor{0,1.7}{$A^*$}
		\end{tikzpicture}
	\end{array},\quad 
\quad \mathbb{E}_A\ket{X} = \ket{\mc{E}_A(X)}.
\ee
The spectrum of $\mathbb{E}_A$ is identical to that of $\mc{E}_A$ and
the transfer matrix $\mathbb{E}_A$ determines the tensor $A$ up to a unitary acting on the physical index, which clearly leaves $\mathbb{E}_A$ invariant. The norm of the MPS is 
\be
\braket{v^{N}(A)|v^{N}(A)} = \tr [\mathbb{E}_A^N].
\ee

Given an MPS defined by the tensor $A^i$, one can construct a new tensor by blocking $q$ sites, such that the corresponding matrices are obtained by the product of original ones, $B^i=A^{i_1}\cdots A^{i_q}$, where $i$ denotes all possible sets of indices $\{i_1,\ldots i_q\}$. Graphically:
\begin{equation}\label{eq:blocking}
  	\begin{array}{c}
		\begin{tikzpicture}[scale=.45,thick,baseline={([yshift=-3ex]current bounding box.center)}]
			\BTensor{0,0}{$B$}
		\end{tikzpicture}
	\end{array}
	= 
  	\begin{array}{c}
	\begin{tikzpicture}[scale=.45, baseline={([yshift=5.5ex]current bounding box.center)}, thick]
		\draw[shift={(0,0)},dotted] (0,0) -- (4,0);
		\ATensor{0,0}{$A$}
		\ATensor{4,0}{$A$}
		\draw [decorate,
    	decoration = {calligraphic brace,mirror}] (0,-0.8) --  (4,-0.8);
		\draw (2,-1.5) node {\scriptsize $q$};
	\end{tikzpicture}
	\end{array}
\end{equation}
The new tensor $B$ has physical dimension $d^q$, the same bond dimension $\chi$ and transfer matrix $\mathbb{E}_B=\mathbb{E}_A^q$.

In general, the map $A \to \mc{V}(A)$ is not injective and different tensors $A,B$ may generate states which are proportional to each other. An example is $B^i = e^{i\phi} X A^i X^{-1}$, where the phase leads to a global phase of the MPS. Another source of redundancies are block upper matrices
\be
A^i=\left(\begin{array}{cc}
A_1^i & A_o^i \\
0 & A_2^i
\end{array}\right),
\ee
which for any choice of $A_o^i$ generate the same MPS. This leads to the following canonical form (CF), 

\begin{defn}\label{defn:CF}
We say that a tensor, $A$, is in a canonical form (CF), if
\be
    A^i=\oplus_{k=1}^r \mu_k A_k^i,
    \ee
where each $A_k$ is a normal tensor and the $\mu_k$ are complex numbers. A tensor $A_k$ is normal if (i) it is irreducible, i.e.~ if $A^i_k$ have no non-trivial invariant subspace, and (ii) it is such that the associated CPM 
\be
\mathcal{E}_{A_k}(X)=\sum_{i=1}^d A_k^i X A_k^{i \dagger}
\ee
has a unique eigenvalue of magnitude and value equal to its spectral radius. The $\mu_k$ are chosen such that the spectral radius equals one. Without loss of generality $|\mu_k|\leq 1$ and at least one of them is equal to one. (This ensures that $\ket{v^{(N)}(A)}$ has bounded norm for $N\to \infty$.)
\end{defn} 
After blocking, for any tensor there exists a tensor in CF generating the same MPS \cite[Proposition 2.4]{cirac2017matrix}. 
 The states generated by $A$ in CF are
\be
\ket{v^{(N)}(A)}=\sum_{k=1}^r \mu_k^N\ket{v^{(N)}(A_k)}.
\ee
Note that different blocks $A_r$ can generate the same states, as we will discuss later.

Next, let us discuss the MPS renormalization group procedure. A single RG step is a map $A \to A'$ defined through the blocking of two sites followed by a projection onto a new effective site:
\begin{align}
\sum_\gamma A^{i_1}_{\alpha \gamma} A^{i_2}_{\gamma \beta} &= \sum_{j=1}^{d'} V^{i_1\, i_2}_{j} A'^j_{\alpha \beta}, \quad d' = \rm{rank}(A^2:\CC^{\chi^2}  \to \CC^{d^2} ) \leq \chi^2, \\
\mathbb{E}_A^2 &= \mathbb{E}_{A'}.
\end{align}
Here $V:\CC^{d'} \to \CC^d \otimes \CC^d$ is an isometry $V^\dagger V = \Id$. $V^\dagger$ maps the two spins of dimension $d$ into an effective spin of dimension $d'$, which is the rank of $A^2$ thought of as a map from the virtual to physical space, or equivalently the rank of $\mathbb{E}^2_A$ as a map from top to bottom. $VV^\dagger$ is a projector $P$ onto the image of $A^2$, $PA^2 = VA' = A^2$. While the projector $P$ is unique, the factorization $A^2 = VA'$ is not. The redundancy can be lifted by requiring $A'$ to be positive, in which case $VA'$ is the polar decomposition of $A^2$. 

The RG step is exact in the sense that we can transform e.g.~ two-site observables $O'=V^\dagger O V $ and $\braket{O'}_{A'} = \braket{O}_{A}$.
Two subsequent RG steps $\mathbb{E}_A^2 = \mathbb{E}_{A'}, \mathbb{E}_{A'}^2 = \mathbb{E}_{A''}$ are equivalent to a single step blocking four sites $\mathbb{E}_A^4 = \mathbb{E}_{A''} $. Indeed, it can be easily seen that $\mathbb{E}_A^4$ and $\mathbb{E}_{A''}$ have the same spectrum. 

For a normal tensor, $A$ the spectral radius of the transfer matrix equals one and the associated eigenvector is unique
\begin{equation}
	\mathbb{E}_A = |R)(L| + S=
  	\begin{array}{c}
		\begin{tikzpicture}[scale=.4,thick,baseline={([yshift=1ex]current bounding box.center)}]
			\draw[shift={(-0.3, 0)}] (-1, 1) -- (-0.5, 1) -- (-0.5, -1) -- (-1, -1);
			\draw[shift={(0, 0)}] (+1, 1) -- (+0.5, 1) -- (+0.5, -1) -- (+1, -1);
			\filldraw[color=black, fill=white, thick](-0.8, 0) circle (0.5);
			\draw (-0.8, 0) node {\scriptsize R};
            \filldraw[color=black, fill=white, thick](0.4, 0) circle (0.5);
			\draw (0.4, 0) node {\scriptsize L};
		\end{tikzpicture}
    \end{array}
    + S
	,
	\label{eq:Ek_decomp}
\end{equation}
where $S$ has spectral radius strictly smaller then one and  $|R,L)$ are the right and left singular vectors of $\mathbb{E}_A$, corresponding to the subspace of eigenvalue one,
which satisfy $(L|R) =1$. The transfer matrix of a RG fixed point satisfies 
\be
\lim_{n \to \infty}\mathbb{E}_A^n = |R)(L| =
  	\begin{array}{c}
		\begin{tikzpicture}[scale=.4,thick,baseline={([yshift=1ex]current bounding box.center)}]
			\draw[shift={(-0.3, 0)}] (-1, 1) -- (-0.5, 1) -- (-0.5, -1) -- (-1, -1);
			\draw[shift={(0, 0)}] (+1, 1) -- (+0.5, 1) -- (+0.5, -1) -- (+1, -1);
			\filldraw[color=black, fill=white, thick](-0.8, 0) circle (0.5);
			\draw (-0.8, 0) node {\scriptsize R};
            \filldraw[color=black, fill=white, thick](0.4, 0) circle (0.5);
			\draw (0.4, 0) node {\scriptsize L};
		\end{tikzpicture}
    \end{array}.
\ee

Using a gauge transformation $A \to X^{-1} A X$ we can achieve 
\begin{equation}
\left| L \right)  =  \sum_{\alpha} \left |\alpha, \alpha\right) ,\quad 
\left| R \right) =  \sum_\alpha \Lambda_{\alpha,\alpha} \left |\alpha, \alpha\right),
\end{equation}
where $\Lambda >0$ is diagonal with trace one. This is called CF II in \citep{cirac2017matrix}.  

The RG fixed point can also be written in terms of projected entangled pair states 
\be
\ket{w}_{R_i L_{i+1}} = \left(\sqrt{\Lambda} \otimes \Id \right)  \sum_{\alpha=1}^{\chi}\ket{\alpha \alpha}_{R_i L_{i+1}},\quad \ket{v^{(N)}(A)} = (\otimes_i V_i) \otimes_{i=1}^N \ket{w}_{R_i L_{i+1}}
\ee
where $V_i$ acts on the sites $L_i,R_i$ (denoting the left and right qudits at each site of the fixed-point lattice). Note that the RG fixed points have zero correlation length $\xi$, where $1/\xi = -\log(|\lambda_2|)$ and $\lambda_2$ is the second-largest eigenvalue of $\mathbb{E}_A$.

We now describe the RG for tensors $\tilde{A}$ in CF with more then one block, $\tilde{A} = \oplus_{k=1}^r \mu_k \tilde{A}_k^i$. Obviously, only the blocks with $|\mu_k| =1$ flow to a non-zero tensor during the RG process. Therefore, we assume for simplicity of notation that $|\mu_k| =1$ for all $k$, i.e.~ compared to \eqref{eq:Ek_decomp} we have
\begin{align}
\mathbb{E}_{\tilde{A}} &= \sum_{k=1}^r |R_k)(L_k| + S = \begin{array}{c}
		\begin{tikzpicture}[scale=.4,thick,baseline={([yshift=1ex]current bounding box.center)}]
            \draw (-0.6,0) -- (0.2,0);
			\draw[shift={(-0.3, 0)}] (-1, 1) -- (-0.5, 1) -- (-0.5, -1) -- (-1, -1);
			\draw[shift={(0.2, 0)}] (+1, 1) -- (+0.5, 1) -- (+0.5, -1) -- (+1, -1);
			\filldraw[color=black, fill=white, thick](-0.8, 0) circle (0.5);
			\draw (-0.8, 0) node {\scriptsize R};
            \filldraw[color=black, fill=white, thick](0.6, 0) circle (0.5);
			\draw (0.6, 0) node {\scriptsize L};
		\end{tikzpicture}
    \end{array} + S, \quad (R_m|L_n) = \delta_{mn}\\
    \lim_{n \to \infty} \mathbb{E}^n_{\tilde{A}} &= \sum_{k=1}^r |R_k)(L_k| = \begin{array}{c}
		\begin{tikzpicture}[scale=.4,thick,baseline={([yshift=1ex]current bounding box.center)}]
            \draw (-0.6,0) -- (0.2,0);
			\draw[shift={(-0.3, 0)}] (-1, 1) -- (-0.5, 1) -- (-0.5, -1) -- (-1, -1);
			\draw[shift={(0.2, 0)}] (+1, 1) -- (+0.5, 1) -- (+0.5, -1) -- (+1, -1);
			\filldraw[color=black, fill=white, thick](-0.8, 0) circle (0.5);
			\draw (-0.8, 0) node {\scriptsize R};
            \filldraw[color=black, fill=white, thick](0.6, 0) circle (0.5);
			\draw (0.6, 0) node {\scriptsize L};
		\end{tikzpicture}
    \end{array} 
\end{align}

Consider two normal tensors $\tilde{A}$ and $\tilde{B}$ and their RG fixed points $A,B$. Then,
\be
|\braket{v^{(N)}(B) | v^{(N)}(A)}|^2 = 0,1
\ee
In the first case,  $A$ and $B$ (and the corresponding MPS) are locally orthogonal (LO)
\be\label{eq:local_ortho}
\sum_{i} A^i \otimes B^{* i} = 0
\ee
In the latter case, $A$ and $B$ are related through a gauge transformation
\be\label{eq:gauge_trafo}
A = e^{i \phi} X B X^{-1} \Rightarrow \ket{v^{(N)}(A)} = e^{i N \phi} \ket{v^{(N)}(B)}
\ee

If the RG fixed points $A_k$ of the $\tilde{A}_k$ are all pairwise locally orthogonal, then the RG fixed point $A$ of $\tilde{A}$ is simply
\be\label{eq:rfp_lo}
A = \bigoplus_{k=1}^r \mu_k A_k^i, \quad
\ket{v^{(N)}(A)} = \bigoplus_{k=1}^r  \mu_k^N \ket{v^{(N)}(A_k)}.
\ee
Consider now the case where some of the $A_k$ are not locally orthogonal and related through a gauge transformation as in Eq.~\eqref{eq:gauge_trafo}. Then, Eq.~\eqref{eq:rfp_lo} can be reduced to
\begin{align}
 \ket{v^{(N)}(A)} & = \bigoplus_{k=1}^g  \left(\sum_{j_k} e^{i \phi_{j_k} N} \right)\ket{v^{(N)}(A_k)} \\
&=\bigoplus_{k=1}^g  \alpha^{(N)}_k \ket{v^{(N)}(A_k)}
 , \quad g < r.
\end{align}
Here we collected the proportional MPSs and combined the phases from the $\mu_k$ and the gauge transformation Eq.~\eqref{eq:gauge_trafo}. With some abuse of notation, the $A_k$ left in the decomposition are a basis of normal tensors. Finally, note that the phases $\phi_{j_k}$ and thus weights $\alpha_k$ can be obtained from the spectral decomposition of the transfer matrix.

As an example, consider an MPS of bond dimension $\chi=3$ generated by $A$ in CF whose non-zero elements are 
\be 
A^0 = \begin{pmatrix}
    1 & 0 & 0 \\
    0 & 0 & 0 \\
    0 & 0 & 0 \\
\end{pmatrix}, A^1 = \begin{pmatrix}
    0 & 0 & 0 \\
    0 & e^{i\phi} & 0 \\
    0 & 0 & e^{-i\phi} \\
\end{pmatrix}.
\ee
This generates the states
\be
\ket{v^{(N)}(A)} = \ket{0}^{\otimes N} + e^{i\phi N} \ket{1}^{\otimes N} + e^{-i\phi N} \ket{1}^{\otimes N}= \ket{0}^{\otimes N}+2\cos(\phi N)\ket{1}^{\otimes N}\,.
\ee


\section{The Main Theorem}\label{sec:main_theorem}
In this section we prove our main result, namely Theorem~\ref{th:main_result} in the main text. We begin with a few preliminary lemmas.

\begin{lem}
\label{lemma:mut_info_stab}
    Consider a pure stabilizer state $\ket{S}$ and a tripartition $\Lambda = A \cup C \cup B$. Then the mutual information of $\rho_{AB} = \mathrm{Tr}_C \ket{S} \bra{S}$ is $I_{A,B}[S] =  S(\rho_A) + S(\rho_B) - S(\rho_{AB})=r $ for some positive integer $r$, where $S(\rho_R) = -\tr{\rho_R \log_2(\rho_R})$.
\end{lem}
\begin{proof}
    This follows immediately from Lemma \ref{lemma:stab_entanglement_spectrum}
\end{proof}

\begin{figure}[!htb]
    \centering
    \includegraphics[width=1\linewidth]{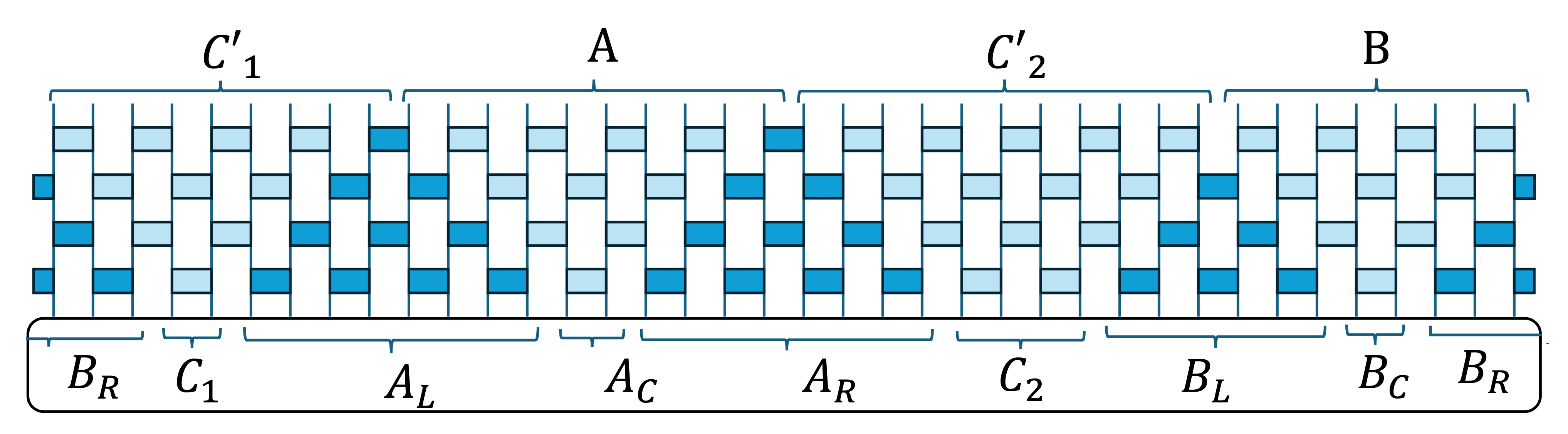}
    \caption{Graphical proof of Lemma~\ref{lemma:main_lemma_tripartion}. The figure shows a QC of depth $D=4$. For each output region, we identify a backward causal cone. For $C_1$ and $C_2$, the gates are removed by taking the trace, while for $A$ and $B$ the gates can be removed by the unitaries $U_A$ and $U_B$. The size of the causal cones depends on whether the boundaries fall in between two unitaries or not. The two cases are displayed for $A/B$. Region $A$ covers $2D+2$, whereas region $B$ covers $2D$ sites. Note the periodic boundary conditions.}
    \label{fig:tripartition}
\end{figure}

\begin{lem}\label{lemma:main_lemma_tripartion}
    Let $\ket{\psi
    }\in \HH_N$,
    $Q_D$ a shallow QC of depth $D$ and $\rho = Q_D \ket{\psi} \bra{\psi} Q^\dagger_D$. Let $A\cup C'_1 \cup C'_2 \cup B = \Lambda$ a partition as shown in Fig.~\ref{fig:tripartition},
    such that each region is of size $|R|\geq 2D+2$.
    Then there exists local unitaries $U_{A/B}$ and a parition of $\Lambda$ into disjoints regions $C_{1/2},A_{L,C,R},B_{L,C,R}$, such that
    \be \label{eq:sigma_ab}
    \sigma_{AB} = U_A \otimes U_B (\tr_{C'
}\rho ) U_A^\dagger \otimes U_B^\dagger = \left( \Id_{A_C} \otimes \mc{E}_{A_L} \otimes \mc{E}_{A_R} \otimes \Id_{B_C} \otimes \mc{E}_{B_L} \otimes \mc{E}_{B_R}
    \otimes Tr_{C} \right)\rho
    \ee 
    where $C'=C'_1\cup C'_2$, $C=C_1\cup C_2$ and the $\mc{E}$ are quantum channels, cf. Fig.~\ref{fig:quantum_channel},
    \be
    \mc{E}_{A_L}: M(A_L) \to M(A\cap  A_L) 
    \ee
    where $M(R)$ denotes the set of density matrices with support in $R$, and analogously for the other regions.
\end{lem}
\begin{proof}
    The proof can be immediately carried out graphically, cf. Fig.~\ref{fig:tripartition}, taking into account the light-cone structure of brickwork quantum circuits. In particular, in the bulk of the regions $C'_1$ and $C'_2$ the unitaries cancel inside the trace. In the bulk of $A$ and $B$ we can choose $U_{A/B}$ to undo the unitaries of the QC. On the boundaries the combination of the circuit and trace results in a quantum channel, cf. Fig.~\ref{fig:quantum_channel}. 
\end{proof}
\begin{figure}
    \centering
    \includegraphics[width=0.25\linewidth]{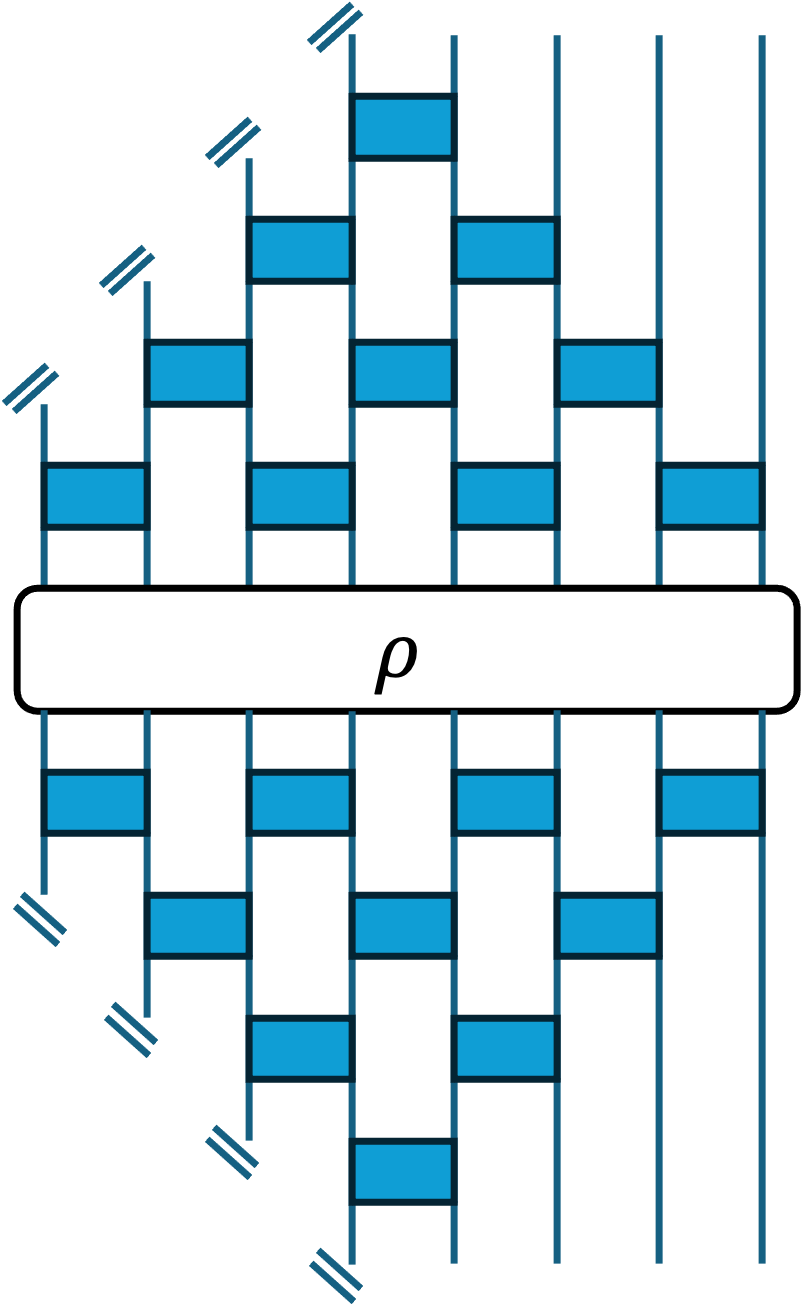}
    \caption{Graphical representation of the quantum channel $\mathcal{E}_A$ appearing in ~\eqref{eq:sigma_ab}. Here we use a tensor-network notation where lower and upper legs correspond to input and output degrees of freedom, respectively. The double lines in the upper and lower left legs denote that they are contracted, \emph{i.e.} the trace is taken over the corresponding qubits.}
    \label{fig:quantum_channel}
\end{figure}
\begin{lem}\label{lemma:mut_info_rfp}
Let \be 
\ket{v^{(N)}(A)} = \dfrac{1}{c_N}\bigoplus_{k=1}^g  \alpha^{(N)}_k \ket{v^{(N)}(A_k)}\,,
\ee
be the RG fixed point of a  translation invariant MPS, where $c_N$ is such that the states are normalized. 
Consider a shallow QC $Q_{D_N}$ [with $D_N=O(\operatorname{polylog}(N)$)] and a tripartition $\Lambda = A \cup C' \cup B = A\cup C'_1 \cup C'_2 \cup B$ such as in Lemma \ref{lemma:main_lemma_tripartion}. Let $\ket{\phi_N} = Q_{D_N} \ket{v^{(N)}(A)}$ Then, the mutual information is
\be
I_{A,B}[\phi_N] = H(\{p_k\}) = -\sum p_k \log_2(p_k)\,,
\ee
where 
\begin{equation}
    p_k = \left |\frac{\alpha_k^{(N)}}{c_N}\right|^2
\end{equation} while $H(\{p_k\})$ denotes the Shannon entropy. In particular, the action of the QC has not changed the mutual information, i.e.~$I_{A,B}[\phi_N] =I_{A,B}[v^{(N)}(A)]$.
\end{lem}

\begin{proof}
    To ease notation, we assume $c_N = 1$ and drop all the dependencies on $N$ in the formulas.
    Using Lemma \ref{lemma:main_lemma_tripartion} and the invariance of $I_{A,B}$ under local unitaries acting on $A/B$ , the mutual information of $\rho_{AB}$ equals the mutual information of $\sigma_{AB}$ given in Eq. \eqref{eq:sigma_ab}. To apply the formula for $\sigma_{AB}$, we need to calculate $\tr_C \ket{v(A)}\bra{v(A)}$. Since $\ket{v(A)}$ is a sum of locally orthogonal terms, see Eq. \eqref{eq:local_ortho}, we have 
    \be
    \tr_C \ket{v(A)}\bra{v(A)} = \sum_k |\alpha_k|^2 \tr_C \ket{v(A_k)}\bra{v(A_k)} =: \sum_k |\alpha_k|^2 \, \rho^k_ {\bar{C}}\, ,
    \ee
     where we use the notation $\tr_R \ket{v(A)} \bra{v(A)} = \rho_{\bar{R}}$. Crucially, because $\ket{v(A_k)}$ is the fixed point of a normal MPS, it has zero correlation length. This implies that taking the trace of $\ket{v(A_k)}\bra{v(A_k)}$ over an interval factorizes the density matrices. In particular,
     \begin{equation}
         \rho_{\bar{C}}^k= \rho^{k}_{A_C\cup A_{L}\cup A_{R}}\otimes \rho^{k}_{B_C\cup B_L\cup B_{R}}\,.
     \end{equation}
    Therefore,
    \begin{align}\label{eq:sigma_ab_mps}
    \sigma_{AB} &= \sum_k |\alpha_k|^2 
    \underbrace{(\Id_{A_C}\otimes\mc{E}_{A_L}\otimes \mc{E}_{A_R})(\rho^{k}_{A_C\cup A_{L}\cup A_{R}})}_{\sigma^k_A}\otimes \underbrace{(\Id_{B_C}\otimes\mc{E}_{B_L}\otimes \mc{E}_{B_R})(\rho^{k}_{B_C\cup B_{L}\cup B_{R}})}_{\sigma^k_B}\\
    &= \sum_k |\alpha_k|^2 \sigma^k_A \otimes \sigma^k_B.
  \end{align}
where $\sigma^k_{A/B}$ are the terms with support on $A/B$. Since the RG fixed points $A_k$ are locally orthogonal, as defined in Eq.~\eqref{eq:local_ortho}, the (matrix) product of the density matrices satisfies $(\rho^{k}_{A_C\cup A_{L}\cup A_{R}})(\rho^{l}_{A_C\cup A_{L}\cup A_{R}}) \propto \delta_{k,l}$ and analogously for $B_C\cup B_{L}\cup B_{R}$. Since $\sigma^k_A$ and $\sigma^k_B$ are obtained by applying to them channels that act as the identity on $A_C$ and $B_C$, we also have
\be
\sigma^k_{A} \sigma^l_{A} =   \sigma^k_{B} \sigma^l_{B}\propto \delta_{k,l}.
\ee
Note that this is true because $A,B \geq 2D+2$, cf.~ Fig.\ref{fig:partition}. Physically, the regions $A$ and $B$ must be larger than the length scale on which the QC can affect the quantum correlations, thus preserving the local orthogonality of the RG fixed points.
    
   Using the orthogonality and standard identities for the von Neumann entanglement entropy $S(\rho) = -\tr \rho \log_2 (\rho)$, we get
   \begin{align}
       S(\sigma_{AB}) &= S\left(\sum_k |\alpha_k |^2\sigma^k_A \otimes \sigma^k_B \right) = H(\{|\alpha_k|^2\}) + \sum_k |\alpha_k |^2 S(\sigma_A^k) + \sum_k|\alpha_k |^2 S(\sigma_B^k) \\
       S(\sigma_A) &= H(\{|\alpha_k|^2\}) + \sum_k |\alpha_k |^2S(\sigma_A^k) \\
       S(\sigma_B) &= H(\{|\alpha_k|^2\}) + \sum_k|\alpha_k |^2 S(\sigma_B^k)  \\
       \Rightarrow I_{A,B}[\phi_N] &= S(\sigma_A) + S(\sigma_B) - S(\sigma_{AB}) = H(\{|\alpha_k|^2\})
   \end{align}
\end{proof}

Let us now recall the definition of LR nonstabilizerness introduced in the main text.
\begin{defn}
A family of states $\{\ket{\psi_N}\}_{N \in \mathbb{N}}$, has short-range (SR) nonstabilizerness if for all $\varepsilon_0>0$ and $\alpha>0$, there exists a local QC $Q_{D_N}$ of depth $D_N =O( {\rm polylog}(N))$ and a stabilizer state $\ket{S_N}$ such that for large enough $N$
\be \label{eq:sm_appendix}
\Delta(Q_{D_N} \ket{\psi_N},\ket{S_N})\leq  \frac{\varepsilon_0}{ N^\alpha} = \varepsilon_N,
\ee
where $\Delta(\ket{\psi},\ket{\phi})=\sqrt{1-|\braket{\psi|\phi}|^2}$ is the trace distance. If the sequence $\{\ket{\psi_N}\}_{N \in \mathbb{N}}$ does not have SR nonstabilizerness, we say that it has LR nonstabilizerness.
\end{defn}
We are finally in a position to prove our main result.
\begin{thm}
A sufficient condition for the RG fixed point of a translation invariant MPS \be 
\ket{v^{(N)}(A)} =\bigoplus_{k=1}^g  \alpha^{(N)}_k \ket{v^{(N)}(A_k)}\,,
\ee
to have LR nonstabilizerness according to Def.~\ref{def:sm} is that 
\begin{equation}\label{eq:main_eq_SM}
    \lim_{N\to \infty} H(\{|\alpha^{(N)}_j|^2\})\notin \mathbb{N}
\end{equation}
where we introduced the Shannon entropy $H(\{p_j\}) = -\sum_j p_j \log_2(p_j)$ and assumed $\sum_j|\alpha^{(N)}_j|^2=1$.
\end{thm}
\begin{proof}
We will show that SR nonstabilizerness implies $   \lim_{N\to \infty} H(\{|\alpha^{(N)}_j|^2\})\in \mathbb{N}$. 
Let \be 
\ket{v^{(N)}(A)} = \bigoplus_{k=1}^g  \alpha^{(N)}_k \ket{v^{(N)}(A_k)}
\ee
be the normalized RG fixed point of a translation invariant MPS. Assume that for all $N$ there is a QC $Q_{D_N}$ such that for large enough $N$
\be\label{eq:assumption_contraditction_SM}
\Delta(\ket{\phi_N}, \ket{S_N}) \leq  \varepsilon_N\,,
\ee
where
\begin{equation}
\ket{\phi_N} = Q_{D_N}\ket{v^{(N)}(A)}\,.
\end{equation}
To calculate the mutual information between distant regions $A$ and $B$ we take a tripartition as in Lemma \ref{lemma:main_lemma_tripartion}. In particular, we take $|A|=|B|$ and
\begin{equation}\label{eq:assumption_size_R}
 4D_N+4\leq   |A\cup B| \leq 8 D_N\,.
\end{equation}
Define now $\rho=\ket{\phi_N}\bra{\phi_N}$ and $\omega=\ket{S_n}\bra{S_N}$ and
\begin{align}
    I_{A,B}[\phi_N] &=  S(\rho_A) + S(\rho_B) - S(\rho_{AB})\\
    I_{A,B}[S_N] &= S(\omega_A) + S(\omega_B) - S(\omega_{AB})\,,
\end{align}
    where $S(\rho)=-\tr [\rho\log_2\rho]$ is the von-Neumann entropy of the density matrix $\rho$ and $\rho_R = \tr_{\bar{R}} \rho $. Then, we can use the triangle inequality
    \begin{equation}\label{eq:triangle}
    |I_{A,B}[\phi_N] - I_{A,B}[
    S_N]| \leq \sum_{R \in (A,B,A\cup B)} | S(\rho_R) - S(\omega_R)|\,,
    \end{equation}
and  bound all terms on the right hand side. To this end, we use the Fannes inequality \citep{audenaert2007sharp}: given two density matrices $\sigma_R$ and $\tau_R$ on the region $R$, with trace distance $\delta = ||\rho-\sigma ||_1 /2$, then
\begin{equation}\label{eq:fannes_SM}
    |S(\sigma_R) - S(\tau_R)|\leq \delta |R| \log_2(2) + H_{\mathrm{bin}}(\delta)\,,
\end{equation}
for any region $R$, where $H_{\mathrm{bin}}(\varepsilon)=-\varepsilon\log_2 \varepsilon -(1-\varepsilon)\log_2 (1-\varepsilon)$.
Now, because the trace distance is contractive, Eq.~\eqref{eq:assumption_contraditction_SM} implies
\begin{equation}\label{eq:fannes_specific}
||\rho_R-\omega_R||_1/2\leq \Delta(\ket{\phi_N},\ket{S_N})=\varepsilon_N\,,
\end{equation}
for any region $R$. Thus, using Eqs.~\eqref{eq:triangle},~\eqref{eq:assumption_size_R}, and ~\eqref{eq:fannes_specific}, it follows that for all $\varepsilon_0>0$ and $\alpha>0$
  \begin{align} \label{eq:mut_info_limit_integer}
    |I_{A,B}[\phi_N] - I_{A,B}[
    S_N]| &\leq 3(\varepsilon_N |A\cup B|) +3H_{\mathrm{bin}}(\varepsilon_N)=O(\varepsilon_N D_N)\,,
\end{align}
so that
\begin{equation}
    \lim_{N\to \infty} |I_{A,B}[\phi_N] - I_{A,B}[
    S_N]|=0\,,
\end{equation}
where we also used $H_{\mathrm{bin}}(0) =0$. The conclusion $\lim_{N\to \infty} H(\{|\alpha^{(N)}_j|^2\})\in \mathbb{N}$ then follows straightforwardly using Lemmas~\ref{lemma:mut_info_stab} and ~\ref{lemma:mut_info_rfp}
\end{proof}

Finally, let us conclude by discussing an example suggesting that the condition of the above theorem is not necessary. Consider the state 
\begin{align}
\label{eq:counter_example_SM}
   \! \ket{\tilde{\phi}_N(t)}=&\alpha_1(t)\ket{00}^{\otimes N/2}+\alpha_2(t)\ket{01}^{\otimes N/2}\nonumber\\
   +&\alpha_3(t)\ket{10}^{\otimes N/2}+\alpha_4(t)\ket{11}^{\otimes N/2}\,,
\end{align}
where $\alpha^2_1(t)=10^{-1}$, $\alpha^2_2(t)=3^{-1/4}10^{-1}$, $\alpha^2_3(t)=t$ and $\alpha_4(t)^2=1-t-10^{-1}-3^{-1/4}10^{-1}$. By inspection, we see that there is a value $t^\ast\simeq 0.023$ for which $H[\{|\alpha_j(t^\ast)|^2\}]=1$. While we cannot show that $\ket{\psi_N(t^\ast)}$ has LR nonstabilizerness, we show that $\ket{\tilde{\phi}_N(t^\ast)}$ cannot be transformed into a stabilizer state exactly, by proving a more general theorem.
\begin{thm}\label{th:theorem_2}
Let \be 
\ket{v^{(N)}(A)} = \bigoplus_{k=1}^g  \alpha^{(N)}_k \ket{v^{(N)}(A_k)}
\ee
be a normalized RG fixed point of a  translation invariant MPS.
Let $Q_{D_N}$ be a shallow QC of depth $D_N=O(\operatorname{polylog}(N))$. Then, for large enough $N$, a necessary condition for $\ket{\phi} = Q_{D_N}\ket{v^{(N)}(A)}$ to be a stabilizer state is $|\alpha_i|^4/|\alpha_j|^4 \in \mathbb{Q}$ for all $i\neq j$. 
\end{thm}

\begin{proof}
Let $\Lambda = A \cup C' \cup B$. Any stabilizer state $\ket{S}$ with reduced density matrix $\rho_{AB} = \tr_{C'} \ket{S} \bra{S}$ satisfies~\cite{sang2021entanglement} 
\be
(\rho_{AB}^{T_A})^2 \propto (\rho_{AB}^{T_A})^4.  
\ee
This is equivalent to 
\be
(\sigma_{AB}^{T_A})^2 \propto (\sigma_{AB}^{T_A})^4
\ee
for $\sigma_{AB} = U_A \otimes U_B \rho_{AB} U_A^\dagger \otimes U_B^\dagger$.
Consider now the state $\ket{\phi}$. For large enough $N$ we can take the tripartition in Lemma \ref{lemma:main_lemma_tripartion}, such that, as in Eq. \eqref{eq:sigma_ab_mps},
\begin{align}
\sigma_{AB} &= U_A \otimes U_B (\tr_{C'
}\ket{\phi} \bra{\phi} ) U_A^\dagger \otimes U_B^\dagger \\
&= \sum_k |\alpha_k|^2 \sigma^k_A \otimes \sigma^k_B
\end{align}
where $\sigma^k_{A/B}$ are locally orthogonal. Thus, $
(\sigma_{AB}^{T_A})^2= \lambda (\sigma_{AB}^{T_A})^4
 $ implies that for all $k$
\be
((\sigma_{A}^{k})^T \otimes \sigma^k_{B})^2=\lambda  |\alpha_k|^4 ((\sigma_{A}^{k})^T \otimes \sigma^k_{B})^4 
\ee 
Finally, Theorem~\ref{th:theorem_2} follows from the following result (whose proof is elementary).
   \begin{lem}
    Suppose $\sigma^2 = \nu \sigma^4$ for $\sigma$ a Hermitian matrix. Then all non-zero eigenvalues $\mu$ of $\sigma$ satisfy $\mu^2 \nu =1$. Furthermore, $\mathrm{Tr}(\sigma) =1 = (D_+ - D_-)|\mu|$ implies $\nu = (D_+ - D_-)^2$, where $D = D_+ + D_- = \mathrm{rank}(\sigma)$ and $D_{\pm}$ are the number of positive (negative) eigenvalues.
\end{lem}
\noindent Now, let $D^k_{\pm}$ be the number of positive (negative) eigenvalues of $(\sigma_{A}^{k})^T \otimes \sigma^k_B$. Then, the Lemma implies $\lambda|\alpha_k|^4  = (D^k_+ - D^k_-)^2 $ for all $k$. Thus, for $i\neq j$
\be
\frac{|\alpha_i|^4}{|\alpha_j|^4} = \frac{(D^i_+ - D^i_-)^2}{(D^j_+ - D^j_-)^2} \in \mathbb{Q}
\ee
which concludes the proof.
\end{proof}

Therefore, we see that there is no shallow QC $Q_{D_N}$ such that $Q_{D_N}\ket{\tilde{\phi}_N(t^\ast)}$ is a stabilizer state. Indeed, we constructed $\ket{\tilde{\phi}_N(t^\ast)}$ such that $|\alpha_1|^4/|\alpha_2|^4 = 1/\sqrt{3}$ is irrational.
\end{document}